\documentclass[]{article}

\usepackage{amsmath}
\usepackage{amssymb}
\usepackage{amsthm}
\usepackage{authblk}
\usepackage{graphicx}
\usepackage{float}
\usepackage{geometry}[1 inch]
\usepackage{color} 
\newtheorem{theorem}{Theorem}
\newtheorem{proposition}[theorem]{Proposition}
\newtheorem{lemma}[theorem]{Lemma}
\newtheorem{corollary}[theorem]{Corollary} 
\theoremstyle{definition}
\newtheorem{definition}[theorem]{Definition}
\theoremstyle{plain}
\newtheorem{example}[theorem]{Example}
\theoremstyle{plain}

\numberwithin{theorem}{section}
\numberwithin{equation}{section}
\newcommand{\f}{\mathbf{f}}
\newcommand{\xixi}{\boldsymbol{\xi}}
\newcommand{\g}{\mathbf{g}}
\newcommand{\z}{\mathbf{z}}
\renewcommand{\c}{\mathbf{c}}
\newcommand{\one}{\mathbf{1}}
\newcommand{\RR}{\mathbb{R}}
\newcommand{\NN}{\mathbb{N}}
\newcommand{\ZZ}{\mathbb{Z}}
\newcommand{\CC}{\mathbb{C}}

\newcommand{\C}{\mathcal{C}}
\renewcommand{\L}{\mathcal{L}}

\renewcommand{\gg}{\tilde{g}}

\newcommand{\B}{{\mathcal B}}
\renewcommand{\H}{\mathcal{B}}
\newcommand{\vol}{\mathrm{vol}}
\newcommand{\twopartdef}[4]
{
	\left\{
	\begin{array}{ll}
		#1 & \mbox{if } #2 \\
		#3 & \mbox{if } #4
	\end{array}
	\right.
}

\newcommand{\twopartdefotherwise}[4]
{
	\left\{
	\begin{array}{ll}
		#1 & \mbox{if } #2 \\
		#3 & \mbox{otherwise} #4
	\end{array}
	\right.
}

\newcommand{\BV}{\B_V}
\newcommand{\PhiV}{\Phi}
\newcommand{\PsiE}{\Psi}
\newcommand{\admissible}{(G, \BV, \PhiV, \PsiE)}

\usepackage{soul}

\title{Stable Phase Retrieval from Locally Stable and Conditionally Connected
 Measurements}

\author{Cheng Cheng, Ingrid Daubechies, Nadav Dym and Jianfeng Lu}

\date{}

\affil[]{Department of Mathematics, Duke University, Durham, North Carolina 27708}

\begin{document}
	
	\setstcolor{red}

	\maketitle

\abstract{
This paper is concerned with locally stable phase retrieval for frames in infinite-dimensional or finite-but large-dimensional Banach spaces. In infinite-dimensional spaces, this phase retrieval
is known to never be uniformly stable. In finite-dimensional spaces, frames that allow phase retrieval are stable, with a finite stability constant; yet when one considers nested hierarchies of finite-dimensional approximation spaces these constants tend to infinity as the dimension grows, possibly suffering a ``curse of dimensionality", i.e. growth may be exponential
in the dimension. Thus, even in finite- but large-dimensional spaces, the guaranteed uniform stability may be too weak to be useful. For these reasons, several recent papers have focused instead on local stability results, studying the extent to which phase retrieval is possible in the neighborhood of particular signals, characterized by special ``connectivity" properties for their measurements. In particular, Grohs and Rathmair use the ``connectivity" of the signal in measurement space to measure the stability of the Gabor phase retrieval in an infinite dimensional setting.

To study the local stability of phase retrievable signals, we introduce the notion of ``locally stable and conditionally connected" (LSCC) measurement scheme associated with frames: to every signal $f$, we associate a corresponding weighted graph $G_f$, defined by the LSCC measurement scheme, and show that the phase retrievability of the signal $f$ is determined by the connectivity of $G_f$. We then characterize the phase retrieval stability of the signal $f$ by two
measures that are commonly used in graph theory to quantify the connectivity of the graph: the Cheeger constant of $G_f$ for real valued signals, and the algebraic connectivity of $G_f$ for complex valued signals.

 We next use our results to study the stability of two phase retrieval models that can be cast as LSCC measurement schemes, and focus on understanding for which signals the ``curse of dimensionality" can be avoided. The first model we discuss is a finite-dimensional model for locally supported measurements such as the windowed Fourier transform. For signals ``without large holes", we show the stability constant exhibits only a mild polynomial growth in the dimension, in stark contrast with the exponential growth of the uniform stability constants; more precisely, in $\RR^d$ the constant grows proportionally to $d^{1/2}$, while in $\CC^d$ it grows proportionally to $d$. We also show the growth of the constant in the complex case cannot be reduced, suggesting that complex phase retrieval is substantially more difficult than real phase retrieval. The second model we consider is an infinite-dimensional phase retrieval problem in a principal shift invariant space. We show that despite the infinite dimensionality of this model, signals with monotone exponential decay will have a finite stability constant. In contrast, the stability bound provided by our results will be infinite if the signal's decay is polynomial.}

	\section{Introduction}
	Phase retrieval considers the problem of recovering  a signal $f$,  up to a global phase, from the magnitudes of its linear measurements, such as the Fourier transform, frame measurements, and the short-time Fourier transform.  Phase retrieval has many applications in applied physics and engineering where the phase information is lost. Examples include optics, speech recognition and  X-ray crystallography \cite{fienup1978, phaseretrievalreview}.  
	 Many algorithms have been developed to reconstruct a signal, up to a global phase, from the magnitudes of its linear measurements,  prominent examples being alternating projection algorithms \cite{fienup1982phase, gerchberg1972}, and convex relaxation algorithms such as the celebrated PhaseLift algorithm \cite{candes2015phase}. There has also been much interest in theoretical study of phase retrieval, focusing on understanding how difficult a given phase retrieval problem is, in two aspects:  (i) phase retrievability and (ii) stability of phase retrieval. 

For a Banach space $\B$ and a set of linear measurements $\Phi$ in the dual of $\B$, $\Phi$ does \emph{phase retrieval} for $\B$ if  every signal $f \in \B$ is determined, up to a global phase, from it phaseless measurements $|\Phi(f)|:= \big(|\phi(f)|\big)_{\phi \in \Phi}$. In \cite{BCD06}, Balan and his collaborators show that  the frame vectors $\Phi$ do phase retrieval for signals in a real finite-dimensional Hilbert space if and only if they satisfy the  complement property. Equivalent conditions for phase retrieval signals in the complex finite-dimensional Hilbert space  are given in \cite{BCMN}. In both cases these conditions are satisfied for almost all $\Phi$ with sufficiently large cardinality. Results in the infinite dimensional settings are also known: In 
\cite{thakur2011reconstruction} the authors show that  real band-limited signals can be determined, up to a sign, from phaseless samples taken more than twice the Nyquist rate.  In \cite{chen2018phase, cheng2019phaseless}, the authors show that not all signals   in the shift-invariant spaces generated by a compact supported function are phase retrievable, and they characterize the phase retrievable signals by the connectivity of an appropriate graph.  In \cite{grochenig2019phase}, the authors show that signals in the shift-invariant space generated by the Gaussian function are phase retrievable if the sampling rate is larger than 2. 


The study of stability of phase retrieval  focuses on the sensitivity of the nonlinear phase retrieval problem to noise corruption of the phaseless measurements, which can be quantified by a stability constant $C>0$ satisfying 
\begin{equation}\label{eq:stable}
\inf_{\xi: |\xi|=1} \|f-\xi g\|_{\B}\leq C \| \, |\Phi(f)| - |\Phi(g)| \, \|_{D}, \ \forall f,g\in \B,  
\end{equation}
where $\|\cdot\|_{D}$ is a suitable norm on the measurement space. For finite-dimensional Hilbert spaces,  phase retrievability implies  the existence of a finite stability constant \cite{BZ16}. In \cite{BCMN},  the authors show that in the real setting the stability constant $C$ is characterized by a matrix condition called the  $\sigma-$strong complement property. Recently, it has been shown that phase retrieval in  infinite dimensional spaces is never uniformly stable, and the stability constant  suffers from the ``curse of dimensionality", as it can grow exponentially in the dimension \cite{alaifari2019stable}, \cite{alaifari2017phase}, \cite{alaifari2019gabor},  \cite{CCD16}. 

It is a folklore in the phase retrieval community   that the main source of phase retrieval instability in high/infinite dimensions is due to a ``disconnectedness" in measurement   space \cite{alaifari2019stable, alaifari2019gabor}. To deal with these connectivity issues,  the authors of \cite{alaifari2019stable} suggest a relaxed notion of phase retrieval up to phase ambiguity per connected component of the measured signal.   In \cite{alexeev2014phase},   the authors suggest  a  design of phase retrieval measurements which is based on expander graphs with good connectivity properties, and they use graph spectral theory to establish  stability for phaseless reconstruction from noisy phaseless measurements. In \cite{Grohs19,grohs2019stable}, the authors study the stability constant $C(f)$  of a fixed signal $f$ for which \eqref{eq:stable} holds for all $g \in \B$, that is
\begin{equation}\label{eq:stable_f}
\inf_{\xi: |\xi|=1} \|f-\xi g\|_{\B}\leq C(f) \| \, |\Phi(f)| - |\Phi(g)| \, \|_{D}, \ \forall g\in \B.  
\end{equation}
They show that the phase retrieval stability of the  continuous Gabor transform can be quantified by an appropriately defined signal-dependent Cheeger constant, which measures the connectivity of phaseless measurements $|\Phi(f)|$.  
	 
Inspired by  \cite{Grohs19,grohs2019stable}, in this paper we  study signal dependent stability constants $C(f)$, which depend  on the connectivity in measurement space. Our results differ from those in \cite{Grohs19,grohs2019stable} in that they consider different phase retrieval models, that the proofs of our results are arguably simpler and more intuitive than in \cite{Grohs19, grohs2019stable}, and that we consider discrete (finite or infinite) measurements in contrast with  the results in \cite{Grohs19, grohs2019stable} that  require a continuum of measurements to achieve stability. 

We focus on a family of phase retrieval models we call \emph{locally stable and conditionally connected (LSCC)} measurement schemes. Intuitively, these are measurement schemes which can be divided into regions where local stable phase retrieval holds, and where global stable phase retrieval of the signal  depends on how well connected these  regions are.  Given a signal $f$, we construct  a graph $G_f$ from the LSCC measurement scheme, and show that  
 connectivity of the graph $G_f$  implies the  phase retrievability of the signal $f$, see Proposition \ref{pro:graph}. We prove stability inequalities of the form \eqref{eq:stable_f}, with a constant $C(f)$  which depends on two common measures of graph connectivity: The Cheeger constant  and the algebraic connectivity.   For real Banach spaces, we characterize the phase retrieval stability constant by the reciprocal of the Cheeger constant  of the graph $G_f$, see Theorem \ref{thm:real}.  For  complex Banach spaces, we characterize the phase retrieval stability constant by   the reciprocal of the {\em algebraic connectivity} of $G_f$, see Theorem \ref{thm:complex}. The relation between the Cheeger constant and algebraic connectivity is captured by the Cheeger inequality \cite{chung2007four,chung1996laplacians}, given in \eqref{ineq:cheeger}, which bounds the algebraic connectivity between the Cheeger constant and the Cheeger constant squared. 

We apply our main results Theorem \ref{thm:real} and   Theorem \ref{thm:complex} to study the phase retrieval stability of two models considered in the literature, that  fulfill the LSCC  assumptions. The first is a model for phase retrieval of signals   in $\CC^d$ or $\RR^d$ from locally supported overlapping phaseless measurements,  such as the discrete windowed Fourier transform. In the context of the $\CC^d$ example, our results can be viewed as
complementary to those in \cite{iwen2020phase, preskitt2019admissible,iwen2019lower}. In \cite{iwen2020phase, preskitt2019admissible}, phase retrieval algorithms for this model
are constructed explicitly, and are proved to have a stability
constant proportional to $d^2$. In  \cite{iwen2019lower}  it is shown that \emph{any} stability constant for this model must grow at least like $d^{1/2} $. We show that this phase retrieval model can be interpreted as an LSCC measurement scheme, and that the phase retrieval stability constant obtained from Theorem~\ref{thm:complex} is linear in $d$. In addition, we show in Proposition \ref{prop:optimal} that the linear dependence of our bounds on $d$ is optimal. The $\RR^d$ example points to the fundamental difference between real and
complex phase retrieval: For signals in $\RR^d$,   Theorem~\ref{thm:real} yields bounds which are proportional to  $d^{1/2}$, which is the optimal rate in the real setting. The difference in scaling of the stability constant in the real and complex settings confirms  that the dependence of the stability constant we provide on  the Cheeger constant and algebraic connectivity respectively, is not an artifact of the proof technique we use but rather due to a fundamental difference between the two settings. 

 The second model we consider as an example of an LSCC measurement scheme is the  principal shape invariant space. Shift-invariant spaces are widely used in  sampling theory, wavelet theory, approximation theory and signal processing \cite{aldroubi2000}.  In \cite{chen2018phase, cheng2019phaseless}, the  phase retrievability of signal in a principal shift invariant space  is  characterized in terms of the connectivity of an appropriate infinite graph, and stability results for reconstruction from noisy phaseless samples are given. Our results complement their results by providing Lipschitz stability results in the sense of \eqref{eq:stable}, which hold for all $1\le p<\infty$ norms. The phase retrieval stability constant is formulated in terms of the Cheeger constant of the  graph proposed in \cite{cheng2019phaseless}. This stability constant may not be finite as for infinite graphs the Cheeger constant can be zero even when the graph is connected. We show that graphs corresponding to a signal  with  monotone exponential decay will have a positive Cheeger constant, while signals with polynomial decay will have a zero Cheeger constant.  

%
  This paper is organized as follows.  In Section \ref{sec:main},  we introduce  LSCC measurement schemes and construct a  signal-dependent graph $G_f$. We then  state our main results on  the  phase retrieval  stability constant both in real and complex settings.  
 We include two applications of our main results   in Section \ref{sec:example}.  The proof of the  main theorems are given in  Section \ref{sec:proof}, and proofs of other propositions stated in the paper are given in Section \ref{sec:add}. In Appendix \ref{app:Cheeger},  we give, for the sake of completeness, a proof of the Cheeger inequality for (in)finite graphs with summable weights. 
 
 \smallskip

 \textbf {Notation}: Throughout the paper we use  $\mathcal B$ to denote a  Banach space over the field $F$ with suitable norm $\|\cdot\|_{\mathcal B}$, where $F=\CC$ or $\RR$.  We denote the dual space by $\mathcal B^*$, i.e., the Banach space of bounded linear functionals $\phi: \mathcal B \mapsto F$. If $\Phi\subseteq \B^*$ is a collection of bounded linear functionals,  we write  
$\Phi(f)=\big(\phi(f)\big)_{\phi \in \Phi}$ and 
$|\Phi(f)|=\big(|\phi(f)|\big)_{\phi \in \Phi}$. For $\phi, \psi:\RR \to \RR_{>0}$,  we say that $\phi \sim \psi $ if there exist $0<a<b $ and $x_0\in \RR$ such that $a\leq \frac{\psi(x)}{\phi(x)}\leq b $ for all $x>x_0$. We denote the cardinality of a finite set $S$ by $\sharp S$.

\section{Main Results}\label{sec:main}

In this section we review some fairly standard phase retrieval definitions, define   {\em locally stable and conditionally connected (LSCC)}  measurement schemes, and study phase retrieval stability constants for these measurement schemes  in both the  real and complex settings.

Frames are redundant systems of vectors in a Banach space. They satisfy the well-known property of perfect reconstruction, that is the linear transformation $\B\ni f\mapsto \Phi(f)$ is injective and hence admits a left inverse. 

\begin{definition}\label{def:frame}
	Let $ \B$ be a Banach space over the field $F=\RR$ or $F= \CC$, and let $\Phi$ be a countable subset of its dual ${\B} ^*$. For $1\le p \le \infty$ and  $0<A \leq B$, we say that $\Phi$ is a \emph{ $p$-frame} on $\B$ with frame constants $(A,B) $, if for all $f \in \B $,
	\begin{equation}\label{eq:frame} 
	A\|f\|_{\mathcal B} \le \|\Phi(f)\|_p\le B\|f\|_{\mathcal B}.
	\end{equation} 
\end{definition}

Phase retrieval frames are frames for which the non-linear mapping $\B\ni f\mapsto |\Phi(f)|$ is injective, up to a global phase, despite the loss of phase information. Many papers have studied the question of when frames do phase retrieval for signals in some space. Readers may refer to \cite{alaifari2017phase, BCD06, BCMN} and references therein. 

\begin{definition}
	Let $ \B$ be a Banach space over the field $F=\RR$ or $\CC$, and let $\Phi$ be a $p$-frame for $1\le p \le \infty $. We say  a signal  $f\in \B$ is \emph{phase retrieval} from  $|\Phi(f)|$ if for every $g \in \B$ satisfying $|\Phi(f)|=|\Phi(g)|$, there exists $\xi \in F, |\xi|=1$ such that $f=\xi g$. We say $\Phi$ is a \emph{phase retrieval frame} on $\B$ if every $f\in \B$ is phase retrievable from $|\Phi(f)|$.    
	\end{definition}
	
Phase retrievability ensures the recovery of a signal $f\in \B$, up to a global phase, from its phaseless measurements. However, reconstructing a signal from its (typically noisy) phaseless measurements can still be very challenging. The difficulty of this non-linear inverse problem can be quantified by the notion of  stability, which is the focus of this paper.
\begin{definition}
	Let $ \B$ be a Banach space over the field $F=\RR$ or $\CC$, and let $\Phi$ be a $p$-frame for $1\le p \le \infty$. 
	We say that $f$ has a phase retrieval stability constant $C(f)$ if
	\begin{equation}\label{eq:PRf}
	\min_{\xi\in F,  |\xi|=1}  \|\Phi(f)-\xi\Phi(g)\|_p \leq C(f) \| \, |\Phi(f)|-|\Phi(g)| \,  \|_p, \,  \forall g\in \B.  
	\end{equation}
	We say that $\Phi$ is a \emph{stable phase retrieval frame} if there exists a global phase retrieval stability constant $C_{\B}>0$ such that 
	\begin{equation}\label{eq:PRframe}
	\min_{\xi\in F,  |\xi|=1}  \|\Phi(f)-\xi\Phi(g)\|_p \leq C_{\B} \| \, |\Phi(f)|-|\Phi(g)| \, \|_p, \,  \forall f, g\in \B.  
	\end{equation}
\end{definition}
We note that  \eqref{eq:PRf} and \eqref{eq:PRframe} focus on  stable recovery of the linear measurements from their magnitudes, up to  a global phase. The assumption on $\Phi$ being a $p$-frame implies the stable recovery of the signal itself from the linear measurements. 

The phase retrieval stability constant   $C(f)$ measures  the robustness  of reconstructing  a signal in a noisy setting from its phaseless measurements. Assume that the phaseless measurements are corrupted by noise $\eta$, i.e., 
\begin{equation}\label{defn:noise}
z=|\phi(f)|+\eta\geq 0,  \phi\in \Phi,  
\end{equation}
and assume that $\hat f$ is a reconstruction to the signal  $f$ which is obtained  by solving 
\begin{equation*}
\hat{f}= \mathrm{argmin}_{f\in \mathcal B} \| \ |\Phi(f)|-z\ \|_p. 
\end{equation*}
Then the error of the reconstruction is proportional to  the  phase retrieval stability constant $C(f)$ and the magnitude of the noise, as 

	\begin{eqnarray}
	\inf_{\xi\in F, |\xi|=1}\| \Phi(f)-\xi \Phi(\hat f)\|_p &\le &C(f) \| \, | \Phi(f)|-| \Phi(\hat f)| \, \|_p \nonumber\\
	&\le& C(f) \left[  \| \, |\Phi(f)|-z \, \|_p+\|\, z-|\hat \Phi(f)| \, \|_p  \right] \nonumber\\
	&\le&2C(f)\| \, |\Phi(f)|-z \, \|_p=2C(f)\| \eta\|_p. 
	\end{eqnarray}
\subsection{LSCC measurement scheme}	

In this subsection, we define LSCC measurement schemes. To motivate this definition, we start this subsection with a toy example: 

\begin{example}\label{ex:toy}
Consider the Banach space  $\B=\RR^4$ endowed with the $\ell^2$ norm. We think of elements of $\RR^4$ as functions from $\{1,2,3,4\}$ to $\RR$. For $1\leq k,\ell\leq 4$ we define ${\bf e}_k$ to be the function taking $k$ to $1$ and $\ell \neq k$ to zero. 
We set $\Phi_{k}=\{{\bf e}_k, {\bf e}_{k+1}, {\bf e}_k+{\bf e}_{k+1}\}$ for $1\le k\le 3$, and $\Phi=\cup_{k=1}^3\Phi_k$. We now consider the phase retrieval problem for this measurement scheme, that is, the problem of reconstructing a signal $f \in \B$, up to a global phase, from the phaseless measurements $|\Phi(f)|$. 

The measurement scheme is ``locally phase retrieval": the value of $\Phi_k(f) $ depends only on the coordinates indexed by $(k,k+1) $. Thus effectively we can think of $\Phi_k$ as functionals defined on the  subspace $\B_k$  which contains  functions $f\in \B$ supported on $\{k,k+1\}$. The fact that $\B_k$ is of  dimension $2$, and any two vectors in $\Phi_k$ are linearly independent, implies that $\Phi_k$ is a (stable) phase retrieval frame for  $\B_k$, due to the complement property \cite{BCD06}. It follows that if $f,g\in \B $ have the same phaseless measurements, then there exist  unimodular constants $\xi_k, k=1,2,3$ such that $f(k)=\xi_kg(k) $ and $f(k+1)=\xi_kg(k+1)$. Now note that
\begin{equation}\label{eq:toy}
\xi_kg(k)=f(k)=\xi_{k-1}g(k), \text{ for } k=2,3 .\end{equation}
and so  $\xi_k=\xi_{k-1} $ if $f(k)\neq 0$. This is the ``conditional connectivity" property -- the different local regions on which phase retrieval is guaranteed can be connected,  but this is conditioned on the sparsity pattern of $f$. For example,  consider the functions
\begin{equation}\label{eq:fj}f_0=\{1,2,3,4 \}, \quad f_1=\{1,2,0,1\} \text{ and } f_2=\{1,8,0,0\} .\end{equation}
We see that $f_0$ is phase retrievable from $|\Phi(f_0)|$, since for $g$ with  the same phaseless measurements as $f_0$, \eqref{eq:toy} implies that $\xi_1=\xi_2=\xi_3 $. In contrast, 
$f_1$ is not phase retrieval as $g=\{1,2,0,-1\}$ will have the same phaseless measurements, i.e., $|\Phi(f_1)|=|\Phi(g)|$. Finally $f_2$ is phase retrievable since it is in $\B_1$. 
\end{example}
\bigskip 

The structure of the measurement scheme of the toy example, where local phase retrieval is known, and global phase retrieval depends on the signal, occurs in several models for phase retrieval, such as the examples discussed in Section~\ref{sec:example}. These examples satisfy the assumptions of {\em  Local Stable and Conditionally Connected} (LSCC)  measurement scheme which we will now define:

\begin{definition}\label{def:conditions}
 Let $G=(V,E) $ be a unweighted graph with bounded degree $D:=Deg(G)$.  Let $$\B_V=(\B_v)_{v\in V}, \quad   \Phi=(\Phi_v)_{v\in V} \text{ and } \Psi=(\Psi_{u,v})_{(u,v)\in E}, $$
 where  $\B_v$ is a finite dimensional subspace of $\B$ for each $v\in V$, and   $\Phi_v$ and $\Psi_{{u,v}}$ are finite subsets of $\B^*$ for all $v\in V $ and $(u, v)\in E$.  
For $1\le p< \infty$, we say that $(G,\BV, \PhiV,\PsiE) $ is a   \emph{local stable and conditionally connected} (LSCC) measurement scheme for $\B$  
 with constants $(p,D,C_0,C_1)$  if the following assumptions  are satisfied:
 \begin{enumerate}
 	\item  \textbf{Local phase retrievability}: For each vertex $v\in V$,
$\Phi_v$ is a $p$-phase-retrieval frame for $\B_v$ with frame constants $(A,B) $ and a phase retrieval stability constant $C_0>0$ which are independent of $v$, i.e., 
\begin{equation}\label{const:local}
\min_{\xi, |\xi|=1} \|\Phi_v(f) -\xi\Phi_v(g)\|_p \le C_0  \|\ |\Phi_v(f)| - |\Phi_v(g)|\ \|_p, \quad \forall f,g \in \B_v.
\end{equation} Moreover, there exists a projection $P_v$ onto $\B_v$ and 
\begin{equation}\label{eq:proj}
\phi \circ P_v=\phi, \quad \forall \phi \in \Phi_v.\end{equation}

\item \textbf{Conditional global connectivity}: There exists $C_1>0$ such that for every edge $({u,v})$ in $E $,
\begin{equation}\label{eq:edges_dominated}
\|\Psi_{u, v}(f)\|_p \leq C_1 \|\Phi_v(f)\|_p \text { and }  \|\Psi_{u, v}(f)\|_p \leq C_1 \|\Phi_u(f)\|_p.
\end{equation}

\item \textbf{Exhaustion} The norms $f \mapsto \|f \|_\B$ and $f \mapsto \left[\sum_{v \in V} \|P_v f\|_\B^p \right]^{\frac{1}{p}} $ are equivalent.  
 \end{enumerate}

\end{definition}
\medskip

We note that with a slight abuse of notation, we will use $\Phi$ to denote both the sequence $\Phi=(\Phi_v)_{v\in V} $ and the union of this sequence $\Phi=\cup_{v\in V} \Phi_v $.

The uniform phase retrieval stability constant in \eqref{const:local}  has been studied in \cite{alaifari2019stable, BZ16, BCMN}. 
 The independence of the constants $A,B,C_0 $ on $v\in V$ is typically due to the symmetric structure of the measurement scheme. 

The functionals $\Psi_{u, v}$ should not be interpreted as  measurements, but rather as information that can be stably  inferred  from the local measurements $\Phi_v$ and $\Phi_u$. Give a signal $f\in \B$, the functionals $\Psi_{u, v}$ can ``glue'' the phase information of the local measurements $\Phi_u(f)$ and $\Phi_v(f)$ to ensure the  consistency of the phase on the edge $(u, v)\in E$, providing that $\Psi_{u, v}(f)\neq 0$.

The exhaustion assumption together with the assumption that each $\Phi_v$ is a frame with constants independent of $v$ implies that $\Phi=\cup_{v\in V} \Phi_v $ is a frame for $\B$. Thus stable recovery of the measurements implies stable recovery of the signal. 


We now return to the toy example, and explain how it can be interpreted as an LSCC measurement scheme.  Let $G=(V,E) $ where 
$$V=\{1,2,3\} \text{ and } E=\{(1,2), (2,3)\}. $$
For $k\in V$,  we take  $\B_k$ and $\Phi_k$ as described in the example and define $\Psi_{(1,2)}=\{\delta_2\} $ and $\Psi_{(2,3)}=\{\delta_3\}$. The local phase retrievability  assumption  holds as  the complement property is satisfied by the construction of $\Phi_k, k\in V$, and \eqref{eq:proj} holds as  $\Phi_k$ depends only on the $k$th and $k+1$th coordinates of $f$ in the toy example.  The  conditional global connectivity  assumption follows from the fact that $\Phi_k, k=1, 2$ are frames for $\B_k$ with some constants $(A,B)$,  and 
$$\|\Psi_{k,k+1}(f)\|_p \leq \|P_kf\|_p \leq A^{-1}\|\Phi_k(P_k f)\|_p=A^{-1}\|\Phi_k( f)\|_p. $$
A similar argument shows that $\|\Psi_{k,k+1}(f)\| \leq A^{-1}\|\Phi_{k+1}( f)\|_p$ so that \eqref{eq:edges_dominated} holds with $C_1=A^{-1} $. Finally the exhaustion property follows from the fact that 
$$\|f\|_\B^2\leq \sum_{k \in V} \|P_kf\|_\B^2 =f(1)^2+2f(2)^2+2f(3)^2+f(4)^2\leq 2 \|f\|_\B^2. $$

\begin{figure}[t]
	\begin{center}
	\includegraphics[width=112mm]{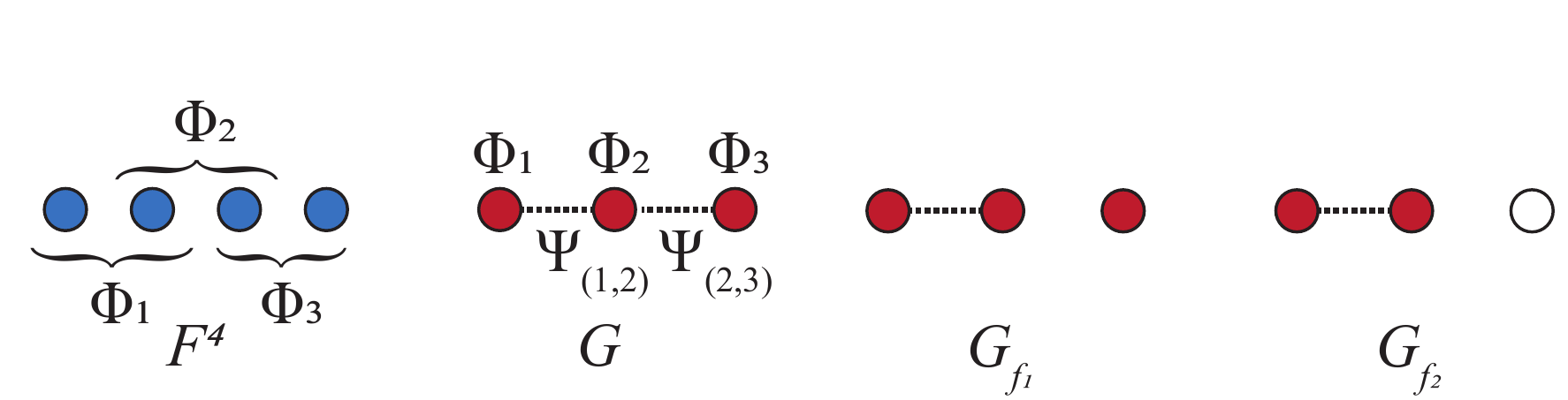}
	\caption{Illustration of the toy examples.}	\label{fig:PR}
	\end{center}
\end{figure}

\subsection{Connectivity and phase retrievability} 
In this subsection, we provide a sufficient condition for the phase retrievability of a  signal $f\in \B$, via the connectivity of a graph $G_f$ induced by the signal $f$ and the LSCC measurement scheme, see Proposition \ref{pro:graph}.  

For a signal $f\in \B$, we define the non-negative weight  $w$  associated  with an  LSCC measurement scheme $\admissible$  as 
\begin{equation}\label{eq: graphweight} w_v=\|\Phi_v(f)\|_p^p, \  \forall v\in V \ \text{ and } \ w_{u,v}=\|\Psi_{u,v}(f)\|_p^p, \ \forall (u, v)\in E, 
\end{equation}
where $1\le p<\infty$. Now use these weights to  define a weighted graph $G_f: =(V_f,E_f, w)$ for a signal $f\in \B$,  where 
\begin{equation}\label{def:graph}V_f=\{v \in V| \, w_v >0 \} \  \text{ and }  \ E_f=\{(u, v)\in V_f \times V_f| \, w_{u, v}>0 \}. 
\end{equation}  

We now return once again to our toy example. Figure~\ref{fig:PR} shows the graphs $G_{f_0}=G$, $G_{f_1}$ and $G_{f_2} $  induced by the functions $f_0,f_1,f_2$ in \eqref{eq:fj}. We notice that the connectivity of the induced graph $G_f$ coincides with the phase retrievability of the signal, as the graphs $G_{f_0}, G_{f_2}$ are connected, while $G_{f_1}$ is not. The next proposition shows that the  connectivity of $G_f$ is a sufficient condition for phase retrievability of the signal $f$:
\begin{proposition}\label{pro:graph}
	Let  $\admissible$ be an LSCC measurement scheme, and let $f$ be a signal in $\B$. If the graph $G_f$ associated with the signal $f$ is connected, then the signal $f$ is phase retrievable from $|\Phi(f)|$. 
\end{proposition}
\begin{proof}
	Suppose there exists a signal $g\in \B$ such that $|\Phi(g)|=|\Phi(f)|$. For any vertex $v\in V_f$, we know that $\Phi_v (g)=\Phi_v (P_vg)= \xi_{v} \Phi_v (P_vf)=\xi_{v} \Phi_v (f)$ for some constant $\xi_v\in F, |\xi_v|=1$,  as $\Phi_v$ is a phase retrieval frame for the space $\B_v$. For an edge $(u, v)\in E_f$, we know that
	$$
	\|\Psi_{u, v}(\xi_{v}f-g)\|_p \le C_1  \|\Phi_{v}(\xi_{v}f-g)\|_p=0
	$$
	and 
	$$
		\|\Psi_{u, v}(\xi_{u}f-g)\|_p \le C_1  \|\Phi_{u}(\xi_{u}f-g)\|_p=0. $$
 Then for any edge $(u, v)\in E_f$, we have  $\xi_{v}\Psi_{u, v}(f)=\xi_{u}\Psi_{u, v}(f)$. By the definition of the edge set $E_f$  we know that $\Psi_{u, v}(f)$ is non-zero and so   $\xi_v=\xi_u$. Since $G_f$ is connected, there exists a unimodular constant $\xi$ such that $\xi_{v}=\xi$ for all $v\in V_f$, so $\Phi_v(f)=\xi \Phi_v(g) $ for all $v\in V_f$. This equality holds for $v \in V \setminus V_f$ as well, as $\Phi_v(f)=0=\xi \Phi_v(g)$. We conclude that $\Phi(f)=\Phi(\xi g)$. Since   $\Phi$ is a frame,  this implies that $f=\xi g$.
\end{proof}


\subsection{Stable phase retrieval in real Banach spaces}
While phase retrievability of a signal $f$ with respect to an LSCC measurement scheme is determined by the connectivity of the graph $G_f$, the phase retrieval stability constant of $f$ is related to constants which measure ``how connected" the weighted graph  $G_f=(V_f, E_f, w)$ in \eqref{def:graph} is. In the real setting,  we characterize the stability constant through a popular measure of graph connectivity known as the Cheeger constant.   To define the Cheeger constant we first define the boundary $\partial S$ of a set $S \subseteq V_f$ as 
$$\partial S=\{(u,v) \in E_f| \ u \in S \text{ and } v \not \in S \}.$$
The Cheeger constant is then defined by 
\begin{equation}\label{eq:cheeger} 
C_G(f)= \inf_{ \substack{S\subset V_f \\ \sum\limits_{v \in S} w_v \leq \frac12 \sum\limits_{v \in V_f} w_v }}\frac{\sum_{(u, v) \in \partial S} w_{u, v}}{\sum_{v \in S} w_v}. 
\end{equation}

In the next theorem, we show the phase retrieval stability constant of $f$ is proportional to the reciprocal of  $C_G^{1/p}(f)$:
\begin{theorem}\label{thm:real}
Assume $\B$ is a Banach space over $\RR$, and $(G,\BV, \PhiV,\PsiE) $ is an LSCC measurement scheme  with parameters $(C_0,C_1,D,p), 1\le p<\infty$. Then there exists  a global constant $C_2=C(p,G,C_0,C_1)$ 
 such that for all $f \in \B$,
	\begin{equation}\label{eq:real_result}
	\min_{\xi \in \{-1,1\}}  \|\Phi(f)-\xi\Phi(g)\|_p \leq C_2(1+ C_G^{-1/p}(f)) \| \, |\Phi(f)|-|\Phi(g)| \,\|_p, \ \forall g \in \B. 
	\end{equation}
\end{theorem}
The proof of the above theorem is included in Section \ref{pf:real}, and the constant $C_2$ is specified in \eqref{eq:real_constant}. We remark that the proof of Theorem \ref{thm:real} relies on the partition induced by an assignment of a real-valued phase in $\{-1, 1\}$ per vertex, and it cannot be  adapted to the complex-valued setting. In the next section, we characterize the phase retrieval stability constant in the complex setting via a different, but  related, measure of connectivity, which is known as algebraic connectivity.

\subsection{Stable phase retrieval in complex Banach spaces}
We now discuss algebraic connectivity and its relation to the stability of LSCC measurement schemes in the complex setting. Let $\admissible$ be an LSCC measurement scheme, let  $f$ be a signal in $\B$,  and let $G_f=(V_f, E_f,w)$ be the corresponding weighted graph in \eqref{def:graph}.  The  (possibly infinite-dimensional) adjacency matrix of the graph $G_f$, denoted by $A_{G_f}$, is given by 
$$A_{G_f}(u,v) =\left\{ \begin{array}{cc}
w_{u,v} & \ {\rm  if} \  (u,v )\in E_f; \\
0 &  \ {\rm otherwise},  
\end{array}\right. $$
where $w_{u,v}, (u,v)\in E_f$ is given in \eqref{eq: graphweight}. The Laplacian matrix $L_{G_f}$ of the graph $G_f$ is given by 
\begin{equation}\label{eq: laplacian}
L_{G_f}= D_{G_f}-A_{G_f},
\end{equation} where $D_{G_f}$ is the diagonal matrix with $D_{G_f}(u,u)=\sum\limits_{v\in  V_f} A_{G_f}(u,v)$. The matrices $A_{G_f},D_{G_f}$, and $L_{G_f}$ are  bounded linear operators on the space of graph signals ${\bf z}=(z_v)_{v\in V_f}$ living on the vertex domain  $V_f$, which have bounded weighted norm
 $$ \ell^2(V_f,w)=\{\z| \, \sum_{v \in V_f} w_v |z_v|^2<\infty \}.  $$ 
Let ${\mathcal L}_{G_f}(\z)$ denote 
$${\mathcal L}_{G_f}({\bf z}):={\bf z}^*L_{G_f}{\bf z}= \sum_{(u, v) \in E_f}w_{u, v}|z_u-z_v|^2. $$
We have $\L_{G_f}({\bf 1})=0$ and $\L_{G_f}({\bf z}-c{\bf 1})=\L_{G_f}({\bf z})$, where $\bf 1$ is the unit vector and $c\in \CC$. Now we are ready to define the the \emph{algebraic connectivity} $\lambda_G(f)$. 

\begin{definition}\label{def:alg_con}
Let $\admissible$ be an LSCC measurement scheme. Let $f$ be a signal in $\B$, and let $G_f$ be the weighted graph induced by $f$ as in  \eqref{def:graph}. The  algebraic connectivity of $G_f$, denoted by $\lambda_G(f)$,  is defined as 
\begin{equation}\label{eq:lambda}
\lambda_G(f)=\inf_{\z \neq {\bf 0}, \langle \z,\one \rangle=0} \frac{ \L_{G_f}(\z)}{  \sum_{v\in V_f} w_v |z_v|^2}
\end{equation}
where  $\langle \cdot,\cdot \rangle $ denotes the inner product in $\ell^2(V_f,w)$.
\end{definition}
When $G_f$ is a finite graph, and $S$ is a $|V_f| \times |V_f|$ diagonal matrix with diagonal entries $S(v, v)=w_{v}$, 
 the algebraic connectivity is the second smallest eigenvalue of the normalized Laplacian matrix $S^{-1/2}L_{G_f}S^{-1/2}$. Similarly for infinite graphs,  $\lambda_G(f)$ is the minimum of the spectrum of the normalized Laplacian operator $S^{-1/2}L_{G_f}S^{-1/2} $, restricted to the subspace orthogonal to $S^{1/2}\one$.

The following theorem characterizes the phase retrieval stability constant  in the complex setting, which is related to the reciprocal of $\lambda^{1/2}_G(f)$.

\begin{theorem}\label{thm:complex}
Assume $\B$ is a Banach  space over $\CC$, and $\admissible$ is an LSCC measurement scheme with parameters $(C_0,C_1,D,p)$, $p=2$. Then there exists a  global constant $C_3=C_3(C_0,C_1,D)$  such that for all $f \in \B$,
	\begin{equation}\label{eq:complex_result}
	\min_{\xi\in \CC, |\xi|=1}  \|\Phi(f)-\xi\Phi(g)\|_2 \leq C_3\big(1+\lambda^{-1/2}_G(f)\big)  \| \, |\Phi(f)|-|\Phi(g)| \, \|_2 , \ \forall g \in \B
	\end{equation}
\end{theorem}
The proof of this theorem is given  in Subsection~\ref{sub:thm_complex}, and the constant $C_3$ is specified in \eqref{eq:complex_constant}.

\subsection{The Cheeger inequality}
For a fixed LSCC measurement scheme $\admissible$, we quantify the phase retrieval stability constant of a fixed  signal $f\in \B$ with an LSCC measurement scheme through two measures of the connectivity of the (possibly infinite) graph $G_f$: the Cheeger constant $C_G(f)$ and algebraic connectivity $\lambda_G(f) $. The relation between these measures of connectivity is given by the Cheeger inequality:  
\begin{equation}\label{ineq:cheeger}2C_G(f) \geq \lambda_G(f) \geq \frac{C_G^2(f)}{2D_N},  \end{equation} 
where $D_N=C_1^2D$ is the normalized degree.   
For the sake of completeness,  we prove  \eqref{ineq:cheeger} in the context of  infinite graphs with summable weights in Appendix~\ref{app:Cheeger}, though the proof essentially follows the standard proof (see e.g., \cite{chung2007four,chung1996laplacians}). 

 We note that the Cheeger inequality \eqref{ineq:cheeger} implies that the algebraic connectivity of a graph is zero if and only if the Cheeger constant is zero. This observation is useful for infinite graphs, where the connectivity of the graph does not imply that the algebraic connectivity/Cheeger constant will be positive.

For $p=2$, combining Theorem \ref{thm:real} and the first inequality in  \eqref{ineq:cheeger},  we can obtain a phase retrieval stability constant in the real setting which is dependent on the algebraic connectivity $\lambda_G(f)$ and is proportional to $\lambda^{-1/2}_G(f) $ as in the complex setting. On the other hand, combining Theorem~\ref{thm:complex} and the second inequality in \eqref{ineq:cheeger},  we can obtain a phase retrieval stability constant for the complex setting which depends on the Cheeger constant, but it is proportional to $C_G^{-1}(f)$, in contrast with its dependence on the Cheeger constant shown in Theorem \ref{thm:real}, which   is proportional to $C_G^{-1/2}(f)$. In Example~\ref{ex:finite},  we will see that it is not  possible to do better than $C_G^{-1}(f)$ in the complex vector space $\CC^d$, thus showing that  the different scaling between Theorem~\ref{thm:real} and Theorem~\ref{thm:complex} is not an artifact of our proof technique,  but rather points to an essential difference in phase retrieval stability between the real and complex settings.

\section{Examples}\label{sec:example}
In this section,  we consider two examples of LSCC measurement schemes, and study the phase retrieval stability constant. In the first example, we consider phase retrieval of signals in a finite dimensional vector space $F^d$ (where $F=\CC $ or $F=\RR $) from locally supported measurements, as suggested in \cite{iwen2019lower},  and apply Theorem \ref{thm:complex} to show the dependence of the stability constant on the dimension $d$ and the field $F$.  In the second example,  we apply Theorem \ref{thm:real} to show the stability of  real phase retrieval signals in an infinite dimensional shift-invariant space studied by \cite{chen2018phase, cheng2019phaseless}, and study cases where a nonzero stability constant exists. 

\subsection{Stability in a finite dimensional phase retrieval model}
In this example,  we consider  phase retrieval of signals in a finite dimensional space from locally supported measurements, as formulated in \cite{iwen2019lower}.
 \begin{example}\label{ex:finite}
	 Throughout this example, for $k\in \NN$, we use the notation $[k]=\{0,1,\ldots,k-1\} $ and $\oplus_k, \ominus_k$ for addition and subtraction modulo $k$ in the group $[k]$.
	
	 Let $\H$ be the Hilbert space $F^d$ where $F=\RR$ or $F=\CC $. We consider $\H$ as the space of  functions from $[d]$ to $F$.
	Let $a,L\in \NN$ such that $d=aL $,  and denote the subspace of functions supported on $[2a]$ by  $\H_0$. Let $\Phi_0\subseteq \H_0^*=\H_0$ be a phase retrieval frame for the space $\H_0$, with frame constants $(A,B)$.   
	
	For $k \in [d]$,  we define $S_k:\H \to \H$ to be the operator
	$$S_{k}f(n)=f(n \ominus_d k). $$
	For $\ell=0,1,\ldots,L-1$, we define
	$$\Phi_\ell=\{S_{\ell a}f| \quad f \in \Phi_0\} \text{ and } \B_\ell=\{S_{\ell a}f| \quad f \in \B_0\} .$$
	We study the phase retrievability of signals $f$ in $\B$ from phaseless measurements $|\Phi(f)| $, where $\Phi=\cup_{\ell \in [L] }\Phi_{\ell} $. This phase retrieval problem can be formulated as an LSCC measurement scheme $\admissible$, where  $G=(V,E)$ is a graph defined by
	\begin{equation}\label{eq:graph.ex} V=[L] \text{ and } E=\{(\ell,\ell')| \ \ell \oplus_L 1=\ell' \}, 
	\end{equation}
	the measurements $\Phi_\ell, \ell \in V$,  and subspace $\B_\ell, \ell \in V$ are defined above, and for every edge $(\ell,\ell')\in E$, we take $\Psi_{(\ell,\ell')}$ to be the evaluation functionals at the points in the intersection of the supports of $\H_\ell $ and $\H_{\ell'}$, that is 
	$$\Psi_{(\ell,\ell')}=\{\delta_k| k \in [2a]\oplus_d \ell a \text{ and } k \in [2a]\oplus_d \ell'a \}.$$
\end{example}
	\bigskip

Following \cite{iwen2019lower}, we devote the remainder of this subsection to the  phase retrieval stability constant of  signals $f$ in the set
	\begin{equation}\label{def:Bst}\B_{s,t}=\{f\in \B| \text{ the average of any } a  \text{ consecutive entries of } |f|^2 \text{ is between } s^2 \text{ and } t^2 \},
	\end{equation}
where $0<s \leq t $. Every signal $f \in \B_{s,t} $ is phase retrievable from $|\Phi(f)|$, as the graph $G_f$ in \eqref{def:graph} induced by $f$  has the same edges and vertices as the connected graph $G$ in \eqref{eq:graph.ex}. We are interested in studying constants $C_{s,t} $ such that  
\begin{equation}\label{eq:st_stable}
\min_{\xi \in F, |\xi|=1 } \| \Phi(f)-\xi \Phi(g) \|_2 \leq C_{s,t} \| \  |\Phi(f)|- |\Phi(g)| \  \|_2, \, \forall f,g \in \B_{s,t} .
\end{equation}
 We will focus on the dependence of the constant $C_{s,t}$ on $L$, for fixed  $s,t,a$ and $\Phi_0$. For the purpose of clarity,  we add the dependence on the dimension $L$ and  the field $F$, explicitly to the definition of $C_{s,t} $, and hitherto use the notation $C_{s,t}(F,L) $. 

The results in \cite{iwen2019lower} imply that constants  $C_{s,t}(F,L) $ in  \eqref{eq:st_stable}  grow at least like $L^{1/2} $, for both $F=\RR $ and $F=\CC$. The authors of \cite{iwen2020phase, preskitt2019admissible} suggest phase retrieval algorithms for this model in the complex setting, and prove the algorithms have a stability constant proportional to $d^2$. 
We will now show that our results in Theorem~\ref{thm:real} and Theorem~\ref{thm:complex} yield bounds $C_{s,t}(F,L) $ which are proportional to $L^{1/2}$ when $F=\RR$ and proportional to $L $ when $F=\CC$. Moreover, we show that this asymptotic dependence on $L$ is tight. In particular, this suggests  that the dependence of our stability constant for the real/complex case on the Cheeger constant/algebraic connectivity is not an artifact of the proof technique we use but rather due to a fundamental difference between the two cases.

We now explain how to obtain constant $C_{s,t}(F,L) $ from our results. For fixed $f \in \B_{s,t}$, Theorem \ref{thm:real} for $F=\RR$ and Theorem \ref{thm:complex} for $F=\CC $ provide stability constants which we denote by $C_{s,t}(f,\RR)$ and $C_{s,t}(f,\CC)$. These constants depend on $L$ only through the Cheeger constant when $F=\RR$ or the algebraic connectivity when $F=\CC$. 

In the following Lemma, we show that the Cheeger constant $C_G(f)$ and the algebraic connectivity $\lambda_G(f)$ are bounded below by   the  Cheeger constant $\widehat C_G$ and algebraic connectivity $\widehat \lambda_G$ of the  unweighted graph $G$. The proof is included in Subsection~\ref{proof:lowbound}.

\begin{proposition}\label{prop:lowbound}
	Let $\admissible$ an LSCC measurement scheme as in Example \ref{ex:finite}. Then for any signal $f\in \B_{s,t}$, we have 
	\begin{equation}\label{eq:uw_bound_w}
	C_G(f) \geq \frac{s^2}{2B^2t^2}\widehat C_G(f)  \text{ and }  \lambda_G(f)\geq \frac{s^2}{2B^2t^2}\widehat \lambda_G, 
	\end{equation}
	where $\widehat C_G(f)$ and $\widehat \lambda_G(f)$ are the Cheeger constant and algebraic connectivity of the cyclic graph with $L$ vertices. 
\end{proposition}


 The graph $G$ is  a cyclic graph with $L$ points, for which it is known \cite{fiedler1973algebraic} that 
\begin{equation}\label{eq:eig}\widehat \lambda_G=2(1-\cos(\frac{2\pi}{L})) \sim \frac{4\pi}{L^2} \text{ and } \widehat C_G=\frac{2}{\lfloor L/2 \rfloor} \sim \frac{4}{L}.\end{equation}
Thus in the real case we have $C_{s, t}(\RR,L)\sim C_G^{-1/2}(f) \sim L^{1/2}$, while in the complex case we have  $C_{s, t}(\CC,L) \sim \lambda^{-1/2}_G(f) \sim L $.  

We now discuss the optimality of $C_{s,t}$, in terms of its asymptotic dependence on $L$.  
The argument in \cite{iwen2019lower} (which focuses on $F=\CC$ but applies for the case $F=\RR$ as well), implies that any bound $C_{s,t}(F,L)$ satisfying \eqref{eq:st_stable} will grow at least like $L^{1/2} $ when $F=\RR$ or $F=\CC$. This shows the optimality of our results for the real case. To show optimality in the complex case, we prove the following lemma which shows that in fact that any constant $C_{s,t}(\CC,L)$ must grow at least linearly in  $L$.  
%
\begin{proposition}\label{prop:optimal}
	Let $\admissible$ be the LSCC measurement scheme described in Example~\ref{ex:finite} with $F=\CC $. Let $(A,B)$ be the frame constants of $\Phi_0$ as in \eqref{eq:frame}, and $s,t $ be some numbers satisfying $0<s\le t$. Then there exist $f,g \in \B_{s,t}$ such that
	$$\min_{\xi \in F, |\xi|=1 } \| \Phi(f)-\xi \Phi(g) \|_2 \geq c_{s, t}(\CC,L) \| \  |\Phi(f)|- |\Phi(g)| \  \|_2, \, \forall f,g \in \B_{s,t} $$
	where
	$$c_{s, t}(\CC,L)=\frac{A}{B}\left(1-\cos\left(\frac{4\pi a}{L}\right)\right)^{-1/2}.$$
\end{proposition}
The proof is included in Subsection~\ref{proof:optimal}. Taylor expansion of the cosine function around $0$ shows that  $c_{s, t}(\CC,L)\sim L$.

\subsection{Stability in an infinite dimensional phase retrieval model}
In the previous example, we considered an LSCC measurement scheme for signals in  $\CC^d$ or $\RR^d$, and saw  that even under the assumption that the signals are uniformly bounded away from zero, the stability constant will  deteriorate as the dimension grows. We now consider a real infinite dimension phase retrieval model for which our methodology provides a stability constant which is proportional to $C_G^{-1/p}(f)$, and we will then discuss examples of signals $f$ for which $C_G(f)$ is (or isn't) strictly positive.
\begin{example}\label{ex:sis}
In this example we consider the real infinite dimensional phase retrieval problem discussed in \cite{chen2018phase, cheng2019phaseless}. Fix some $d\in \NN$ and $1\le p<\infty$, let $B_N$ be a continuous function supported in $[0,N]^d$, and define
\begin{equation*}V_p(B_N)=\{f | \quad f(x)=\sum_{k\in \ZZ^d} c_kB_N(x-k) \text{ for some } \c \in \ell^p(\ZZ^d) \}.  \end{equation*}
We assume that $V_p(B_N)$ satisfies the local linear independence property on any open set as defined in \cite{cheng2019phaseless}. This means that  the mapping $\ell^p(\ZZ^d) \ni \c \mapsto \sum_{k\in \ZZ^d} c_kB_N(x-k)  $ is injective, so that we can identify   $V_p(B_N)$  with  $\ell^p(\ZZ^d)$ and endow it with the $\ell^p$ norm 
$$\|f\|_{V_p(B_N)}:=\|\c\|_{\ell^p(\ZZ^d)}, \ f\in V_p(B_N), $$
 where $\bf c$ is the coefficient vector of the signal $f$. We further assume that phaseless measurements $|f(\gamma)|, \gamma\in \Gamma\subset [0, 1]^d$ are sufficient to determine the value of $f\in \B_0\subset V_p(B_N)$, where 
\begin{equation}\label{eq:B0}
\B_0=\{f | \quad f(x)=\sum_{k\in \ZZ^d} c_kB_N(x-k) \text{ where } c_k=0 \text{ if } k \not \in K_0  \} \end{equation}
and 
\begin{equation}\label{Ktheta.def0} K_{0}=\{k\in \ZZ^d| \ B_N(x-k)\neq 0 {\rm \ for \ some  \ } x \in [0, 1]^d\}.  \end{equation} 
We note that this assumption is equivalently to the assumption that the rows of the matrix
\begin{equation}\label{def:Bgamma}{\bf B}_\Gamma:=\big(B_N(\gamma-k)\big)_{\gamma\in \Gamma, k\in K_0} \end{equation} 
are a phase retrieval frame for the space $\RR^{\sharp K_0}$.

Now we are ready to define an LSCC measurement scheme in the space $V_p(B_N)$.  We let the graph $G$ be the lattice $\ZZ^d$, where $V=\ZZ^d$ and $(k ,k')\in E$ if and only if $\|k-k'\|_2=1$. For $\ell \in \ZZ^d$,  we set 
 \begin{equation}\label{eq:phi_l}\Phi_\ell=\{\delta_{\gamma+\ell}| \ \gamma\in \Gamma \} \end{equation} 
and  
$$\B_\ell=\{f \in V_p(B_N)|\quad f(x)=\sum_{k\in \ZZ^d} c_kB_N(x-k) \text{ where } c_k=0 \text{ if } k \not \in K_\ell \},  $$
where $K_\ell=K_0-\ell$.  By the shift-invariance of the space $V_p(B_N)$ and the choice of the set $\Gamma$,  the local phase retrievability assumptions on $\B_\ell, \ell\in \ZZ^d$ are satisfied, cf. \cite{chen2018phase, cheng2019phaseless}. For $(\ell,\ell')\in E $,  we define  $$\Psi_{\ell, \ell'}(f) = (c_k)_{k \in  K_{\ell'}\cap K_{\ell}},$$
 where $\bf c$ is the coefficient vector of the signal $f$. The quadruple $\admissible $ thus defined is an LSCC measurement scheme. 
\end{example}
\bigskip


For a signal $f \in V_p(B_N)$, Theorem \ref{thm:real} gives us a stability result for signals $f\in V_p(B_N) $. When $p=2$,  the  stability constant is explicitly given in the following corollary, which combines our results with the standard method of characterizing the constant $C_0$ in \eqref{const:local} via the  $\sigma$-strong property in  \cite{alaifari2017phase,BCMN}. Proof of the corollary is given in Subsection \ref{sub:cor}.
	
	\begin{corollary}\label{cor:sis}
		Let $V_2(B_N)$ be a real shift-invariant space generated by a continuous compactly supported function $B_N$. Assume that $V_2(B_N)$ is locally independent on all open sets and $\Gamma\subset [0, 1]^d$ is so chosen such that ${\bf B}_\Gamma= \big(B_N(\gamma-k)\big)_{\gamma\in \Gamma, k\in K_0}$  is a phase retrieval frame for $\RR^{\sharp K_0}$, with $K_0$ defined as in \eqref{Ktheta.def0}. 
	Set 
		\begin{equation}\label{def:sigma}  \sigma =\min_{\Gamma'\subset \Gamma} \max \big( \inf_{\|{\bf c}\|_2=1} (\sum_{\gamma\in \Gamma'} |\sum_{k\in K_0}c_kB_N(\gamma-k)|^2)^{1/2}, \inf_{\|{\bf c}\|_2=1} (\sum_{\gamma\in \Gamma\setminus\Gamma'} |\sum_{k\in K_0}c_kB_N(\gamma-k)|^2)^{1/2}\big).\end{equation}
		Then for any $f\in V(B_N)$, we have  
		\begin{equation}
		\min_{\xi\in\{\pm 1\}} \big(\sum_{\gamma\in \Gamma+\ZZ^d}|f(\gamma) - \xi g(\gamma)|^2\big)^{1/2}\le C_{V_2(B_N)} (1+C_G(f)^{-1/2})\Big( \sum_{\gamma\in \Gamma+\ZZ^d} \big| |f(\gamma)|-|g(\gamma)| \big|^2\Big)^{1/2}, \end{equation}
		where $C_{V_2(B_N)} =\max\{2 \sigma^{-1} (\sharp K_0)^{1/2}\lambda_{\max} ({\bf B}_\Gamma), 4\sqrt{2}\sigma^{-2}(\sharp K_0)\lambda_{\max} ({\bf B}_\Gamma)\} \}$, and $\lambda_{\max}({\bf B}_\Gamma)$ is the largest singular value of the matrix ${\bf B}_\Gamma$. 
	\end{corollary}

	In contrast with finite graphs, the Cheeger constant $C_G(f)$  for an infinite  graph $G_f$  may be zero even if $G_f $ is connected. 	Let $\admissible$ be the LSCC measurement defined in Example \ref{ex:sis}. The following proposition shows that if $f$ has a  coefficient sequence $\c $ which decays exponentially and monotonously in $|k|, k\in \ZZ$,   the graph $G_f$ induced by the   signal $f$  has a positive Cheeger constant $C_G(f)$, and thus stable phase retrieval holds for such signal $f$.  In contrast if $\c$ exhibits monotone polynomial decay in $|k|, k\in \ZZ$, then the Cheeger constant $C_G(f)$   will be zero. The proof is included in  Subsection \ref{sub:cheeger_infinite}.


	\begin{proposition}\label{prop:example}
		Let $\admissible$ be the LSCC measurement defined in Example \ref{ex:sis},
		Let $f(x)=\sum_{k\in \ZZ} c_kB_N(x-k) $ where  $\c \in \ell^p(\ZZ)$, and let $G_f$ be the graph induced by $G$ and $f$. 
		\begin{enumerate} 
			\item If there exists some $\beta>0$ such that 
			\begin{equation}\label{prop_ex:1}|c_k|^p=e^{-\beta|k|}, \ \forall k\in \ZZ,  \end{equation}
			then the Cheeger constant of $G_f$ is positive.
			\item If there exists some $\beta>1$ such that 
			\begin{equation}\label{prop_ex:2}|c_k|^p=(1+|k|)^{-\beta}, \  \forall k\in \ZZ, \end{equation}
			then the Cheeger constant of $G_f$ is zero. 
		\end{enumerate} 
	\end{proposition}

\section{Proofs of main theorems}\label{sec:proof}
In this section, we include the proofs for Theorem \ref{thm:real} and Theorem \ref{thm:complex}. We begin with some notation and computations which will be used for both proofs. We use $\f$ and $|\f|$ to denote the sequences
	$\f=\Phi(f) \text{ and } |\f|=|\Phi(f)|  $ in $\ell^p(\Phi)$, so that   \eqref{eq:real_result} and \eqref{eq:complex_result}   can be rewritten in terms of bounding $\| \f-\g \|_{\ell^p(\Phi)}$ by $\|\, |\f|-|\g| \,\|_{\ell^p(\Phi)}$.
	For any $v \in V_f$ we also use $\|\cdot\|_{\ell^p(\Phi_v)} $ to denote 
	$$\|\f\|_{\ell^p(\Phi_v)}=\|\Phi_v(f)\|_p. $$
The proof of both our main theorems relies on the following lemma:

		\begin{lemma}\label{lem:2}
	Let $\admissible$ be an LSCC measurement scheme with parameters $(C_0,C_1,D,p)$, $1\le p< \infty$, and $G_f=(V_f, E_f, w)$ be the weighted graph associated with the signal $f\in \B$.  For $({u,v})\in E_f$, let $\xi_u $ and $\xi_v$ be unimodular constants which minimize $\xi \mapsto \|\f-\xi \g\|_{\ell_p(\Phi_v)} $ and $\xi \mapsto \|\f-\xi \g\|_{\ell_p(\Phi_u)} $ respectively. Then
	 \begin{equation}\label{eq:edgeAndNorm}
	|{\xi}_u-\xi_v|^p w_{{u,v}}\leq 2^{p-1}C_0^pC_1^p\Big(\| \,|\f|-|\g| \,\|_{\ell^p(\Phi_v)}^p+\| \,|\f|-|\g| \,\|_{\ell^p(\Phi_u)}^p\Big). \end{equation}
\end{lemma}
\begin{proof}
	Fix some $({u,v}) \in E_f$. Then by \eqref{eq:edges_dominated} we have
	\begin{align*}
	 \|\bar{\xi}_v\Psi_{{u,v}}(f)-\Psi_{{u,v}}(g)\|_p& =  \|\bar{\xi}_v\Psi_{{u,v}}(f-\xi_v g)\|_p  \le C_1 \|\Phi_v(f-\xi_v g)\|_p \nonumber\\
	 &\leq C_0C_1\|\, |\f|-|\g| \, \|_{\ell^p(\Phi_v)}. 
	\end{align*}
	Similarly, we have 
		\begin{equation*}
	 \|\Psi_{{u,v}}(g)-\bar{\xi}_u\Psi_{{u,v}}(f)\|_p \leq C_0C_1\|\ |\f|-|\g|\ \|_{\ell^p(\Phi_u)}.  		\end{equation*}
Applying the triangle inequality we obtain
	$$  |{\xi}_v-\xi_u|\| \Psi_{{u,v}} (f)\|_p\leq C_0C_1\Big(\| \, |\f|-|\g| \, \|_{\ell^p(\Phi_v)}+\| \, |\f|-|\g| \, \|_{\ell^p(\Phi_u)}\Big). $$	
 We now obtain \eqref{eq:edgeAndNorm} by taking the last equation to the power of $p, 1\le p<\infty$ and using the inequality
\begin{equation*}
(|a|+|b|)^p \leq 2^{p-1}(|a|^p+|b|^p). 
\end{equation*}
\end{proof}

%
%
	\subsection{Proof of Theorem~\ref{thm:real}}\label{pf:real}

In this subsection, we give the proof of Theorem \ref{thm:real}, which provides the phase retrieval stability constant in the real setting. Fix some $f,g \in \B $.  We define the set 
	$$ S_1=\{v \in V_f | \ \|\f - \g\|_{\ell^p(\Phi_v)}\le \|\f+\g\|_{\ell^p(\Phi_v)}\} $$
	and 
	$$ S_2= V_f \setminus S_1= \{v \in V_f | \ \|\f - \g\|_{\ell^p(\Phi_v)}> \|\f+\g\|_{\ell^p(\Phi_v)}\}.$$
Since the inequality we want to prove holds for $g$ if and only if it holds for $-g$, we can assume	without loss of generality that 
\begin{equation}\label{peq: cheeger1}\sum_{v\in S_1}w_v \ge \sum_{v\in S_2}w_v. \end{equation} 
For $v \in V$,  we define $\xi_v=1$ if $v \in S_1$ and $\xi_v=-1$ otherwise, so that $\xi_v$ is a unimodular constant minimizing $ \|f-\xi g\|_{\ell^p(\Phi_v)}$ as required in the conditions of Lemma~\ref{lem:2}.

Sine $\admissible$ is an LSCC measurement scheme, we have 
 \begin{eqnarray}\label{peq: chain1}
  \min_{\xi \in \{-1,1\}} \|\f-\xi \g\|_{\ell_p(\Phi)}^p&\leq&  \|\f- \g\|_{\ell_p(\Phi)}^p \nonumber\\
  &=&\sum_{v\in S_1}\|\f-\g\|_{\ell^p(\Phi_v)}^p+\sum_{v\in S_2}\|2\f-(\f+\g)\|_{\ell^p(\Phi_v)}^p\nonumber\\
  &{\le}& 2^{p-1}\big( \sum_{v\in S_1}\|\f-\g\|_{\ell^p(\Phi_v)}^p  +\sum_{v\in S_2}\|\f+\g)\|_{\ell^p(\Phi_v)}^p\\
  &+&\sum_{v\in S_2}\|2\f \|_{\ell^p(\Phi_v)}^p \big) \nonumber\\
&\leq& 2^{p-1}C_0^p \| \, |\f| -|\g| \,  \|_{\ell^p(\Phi)}^p +2^{2p-1}\sum_{v\in S_2} \| \f \|_{\ell^p(\Phi_v)}^p. \nonumber
\end{eqnarray}
 By \eqref{peq: cheeger1} and  the definition of the Cheeger constant, we have that 
\begin{eqnarray}\label{peq: chain2}
\sum_{v\in S_2} \| \f \|_{\ell^p(\Phi_v)}^p&\le& C_G^{-1}(f) \sum_{(u, v) \in \partial S_2}   w_{u,v}\nonumber \\
  &=& C_G^{-1}(f)2^{-p} \sum_{({u,v}) \in \partial S_2}  |\xi_v-\xi_u|^p w_{u,v}\nonumber \\ 
  &\stackrel{ \text{Lem.~\ref{lem:2}}}{\leq}& 2^{-1}C_G^{-1}(f) C_0^p C_1^p \sum_{({u,v}) \in \partial S_2} \Big(\| \,|\f|-|\g| \,\|_{\ell^p(\Phi_v)}^p+\| \,|\f|-|\g| \,\|_{\ell^p(\Phi_u)}^p\Big)\nonumber \\
    &\le& 2^{-1}D C_G^{-1}(f) C_0^pC_1^p\| \,|\f|-|\g| \,\|_{\ell^p(\Phi)}^p.
\end{eqnarray}
Combining  \eqref{peq: chain1} and \eqref{peq: chain2}, we have 
\begin{align*}
 \min_{\xi \in \{-1,1\}} \|\f-\xi \g\|_{\ell_p(\Phi)}^p& \leq \left[ 2^{p-1}C_0^p+2^{2p-2}DC_G(f)^{-1}C_1^pC_0^p \right]  \| \,|\f|-|\g|  \,\|_{\ell_p(\Phi)}^p \\ 
 &\le C_2^p(1+C_G(f)^{-1}) \| \,|\f|-|\g| \, \|_{\ell_p(\Phi)}^p\\
 &\le C_2^p(1+C_G(f)^{-1/p})^p \| \,|\f|-|\g| \, \|_{\ell_p(\Phi)}^p, 
	\end{align*}
	where 
	\begin{equation}\label{eq:real_constant}
	C_2^p= \max\{2^{p-1}C_0^p, 2^{2p-2}DC_1^pC_0^p\}.
	\end{equation} 
	This completes the proof of Theorem~\ref{thm:real}. \qedsymbol

\subsection{Proof of Theorem \ref{thm:complex}}\label{sub:thm_complex}
In this subsection, we give the proof of Theorem \ref{thm:complex} which provides the phase retrieval stability constant in the complex setting. 
The following technical  lemma is crucial for the proof of Theorem \ref{thm:complex}, as it essentially enables us to replace $\min_{\xi \in \CC, |\xi|=1}\|\f -\xi\g \|_{\ell^2(\Phi)}$ with \\ $\min_{c \in \CC} \|\f -c\g \|_{\ell^2(\Phi)}$. 
	\begin{lemma}\label{lem:toC}
		For all $\f,\g \in \ell^2(\Phi)  $,
		$$\min_{\xi \in \CC, |\xi|=1} \|\f-\xi\g \|_{\ell^2(\Phi)}\leq \sqrt{2} \min_{c\in \CC} \|\f-c\g \|_{\ell^2(\Phi)}+\| \,|\f|-|\g| \,\|_{\ell^2(\Phi)}. $$
	\end{lemma}
	\begin{proof}
		For $\g=\bf 0$ the inequality holds. We can now assume $\g \neq \bf 0$. For fixed $\f,\g\neq 0$  the function $c \in \CC \mapsto \|\f -c \g\|_{\ell^2(\Phi)}^2$ is minimized by $c_*=r_* \xi_* $, where  and $\xi_*$ is a unimodular constant chosen so that $\langle \f, \xi_* \g \rangle $ is real and non-negative, and $r_*$ is  given by
		$$  r_*=\frac{|\langle \f,\g \rangle|}{ \|\g\|_{\ell^2(\Phi)}^2 }. $$
		Define $r_N=\| \f\|_{\ell^2(\Phi)}/\|\g\|_{\ell^2(\Phi)} $, and note that replacing $r_*$ with $r_N$ only costs a constant factor since
		\begin{align}\label{eq:a}
		\|\f-c_* \g\|_{\ell^2(\Phi)}^2=\| \f \|_{\ell^2(\Phi)}^2-\frac{|\langle \f, \g \rangle|^2}{\| \g\|_{\ell^2(\Phi)}^2} \geq \|\f\|_{\ell^2(\Phi)}^2-\frac{\|\f\|_{\ell^2(\Phi)}}{\|\g\|_{\ell^2(\Phi)}} |\langle \f, \g \rangle|=\frac{1}{2} \|\f-r_N \xi_* \g\|_{\ell^2(\Phi)}^2. 
		\end{align}
		We now have
		\begin{eqnarray*}\min_{\xi \in \C, |\xi|=1} \|\f-\xi\g \|_{\ell^2(\Phi)} 
		&\leq & \|\f-\xi_*\g \|_{\ell^2(\Phi)} \leq \|\f-r_N\xi_*\g \|_{\ell^2(\Phi)}+|r_N-1|\|\g \| 
	_{\ell^2(\Phi)} \\
	&\stackrel{\eqref{eq:a}}{\leq}& \sqrt{2}\|\f-c_* \g\|_{\ell^2(\Phi)} 
		+\left| \| \f \|_{\ell^2(\Phi)}-\| \g \|_{\ell^2(\Phi)} \right| \\
		& \stackrel{(*)}{ \leq} & \sqrt{2}\|\f-c_* \g\|_{\ell^2(\Phi)}+ \| \, |\f| - |\g| \, \|_{\ell^2(\Phi)}. 
		\end{eqnarray*}
		 Here $(*)$ follows from the observation  $\f$  and $|\f|$ have the same norm, and so do $\g$ and $|\g|$, and then applying the reverse triangle inequality.
 \end{proof}


\begin{proof}[Proof of Theorem \ref{thm:complex}]
Fix $f,g \in \B $. Let  $\xixi=(\xi_{v})_{v \in V_f} $ be a choice of unimodular constant $\xi_v$ per vertex $v$ which minimizes $\xi \mapsto \|\f-\xi \g\|_{\ell_p(\Phi_v)} $. Then	
	\begin{eqnarray}\label{eq:laplace_bound}
	\L_{G_f}(\xixi)&=&\sum_{({u,v}) \in E_f}|{\xi}_v-\xi_u|^2 w_{{u,v}}\stackrel{\text{Lem. \ref{lem:2}}}{\leq} 2C_0^2C_1^2\sum_{({u,v}) \in E_f} \| \,|\f|-|\g| \,\|_{\ell^2(\Phi_v)}^2+\| \,|\f|-|\g| \,\|_{\ell^2(\Phi_u)}^2\nonumber\\
	&\leq& 2C_0^2C_1^2 D \| \,|\f|-|\g| \,\|^2_{\ell^2(\Phi)}.
	\end{eqnarray}
 Take  $c_0 \in \CC $ such that $\xixi-\bar c_0 \one $ and $\one$ are orthogonal in $\ell^2(V_f,w) $.  We then obtain
	\begin{eqnarray*}
	\min_{\xi \in S^1}  \| \xi\f-\g\|_{\ell^2(\Phi)}^2 &\stackrel{\text{Lem.~\ref{lem:toC}} }{\leq}& 4 \min_{c\in \CC} \|c\f-\g\|_{\ell^2(\Phi)}^2+2\| |\f|-|\g| \|_{\ell^2(\Phi)}^2\\
	&\leq & 4  \|c_0\f-\g\|_{\ell^2(\Phi)}^2+2\| |\f|-|\g| \|_{\ell^2(\Phi)}^2\\
	&=& 4 \sum_{v \in V_f} \|(c_0\f-\bar \xi_v \f)+(\bar \xi_v \f-\g)\|_{\ell^2(\Phi_v)}^2+2\| |\f|-|\g| \|_{\ell^2(\Phi)}^2    \\
	&{\leq}&  8\sum_{v\in V_f}|\xi_v-\bar c_0|\|\f|_{\ell^2({\Phi_v})}^2+2(4C_0^2+1) \| |\f|-|\g| \|_{\ell^2(\Phi)}^2\\
	&\stackrel{\eqref{eq:lambda}}{\leq}& 8\lambda_G^{-1}(f)L(\xixi-\bar{c}_0\one)+2(4C_0^2 +1) \| |\f|-|\g| \|_{\ell^2(\Phi)}^2\\
	&\stackrel{\eqref{eq:laplace_bound}}{\leq}& (16\lambda_G^{-1}(f)C_0^2C_1^2D+8C_0^2+2) \| |\f|-|\g| \|_{\ell^2(\Phi)}^2 \\
	&\leq & \left(4C_0C_1D^{1/2}\lambda_G^{-1/2}+(8C_0^2+2)^{1/2}\right)^2\| |\f|-|\g| \|_{\ell^2(\Phi)}^2,
	\end{eqnarray*}
	By taking the square root of this inequality we obtain
	$$\min_{\xi \in S^1}  \| \xi\f-\g\|_{\ell^2(\Phi)}\leq C_3(1+\lambda_G^{-1/2}) |\f|-|\g| \|_{\ell^2(\Phi)} $$
	where
	\begin{equation}\label{eq:complex_constant}
	C_3=\max \{4C_0C_1D^{1/2}, (8C_0^2+2)^{1/2} \}. 
	\end{equation}
\end{proof}

\section{Additional Proofs}	\label{sec:add}
In this section, we include the proofs for Propositions \ref{prop:lowbound}, \ref{prop:optimal}  and \ref{prop:example}  and Corollary \ref{cor:sis} in Section \ref{sec:example}. 
\subsection{Proof of Proposition~\ref{prop:lowbound}}\label{proof:lowbound}
Throughout the proof we use the notation defined in Example~\ref{ex:finite}. For a signal $f\in  \B_{s,t}$ in \eqref{def:Bst}, we know that the graphs $G_f$ and $G$ have the same vertex set and edge set. 
By the assumption on the LSCC measurement scheme and the definition of weighted graph $G_f$ in \eqref{eq: graphweight}, we have   
	\begin{align*} w_\ell=\|\Phi_\ell(f)\|_2^2=\|\Phi_\ell(P_\ell f)\|_2^2  
	\leq B^2 \|P_\ell f\|^2_{2}\leq  2aB^2 t^2 {\ \rm for \  all } \ \ell\in V
	\end{align*}
	and 
	\begin{align*}w_{(\ell,\ell')} \ge a s^2 {\ \rm for \  all } \ (\ell,\ell')\in E. 
	\end{align*} 
By the definition of the  Cheeger constant in \eqref{eq:cheeger}, we obtain
	\begin{equation*}
	C_G(f)= \inf_{ S\subseteq V }\frac{\sum_{({u,v}) \in \partial S} w_{{u,v}}}{\min \{\sum_{v \in S} w_v, \sum_{v \not \in S} w_v \}} \geq 	\frac{s^2}{2B^2t^2}  \inf_{ S\subseteq V }\frac{\sum_{({u,v}) \in \partial S} 1}{\min \{\sum_{v \in S} 1, \sum_{v \not \in S} 1 \}}= \frac{s^2}{2B^2t^2} \widehat C_G. 
	\end{equation*}
	Similarly,  using an equivalent definition of algebraic connectivity	\begin{align*}
	\lambda_G(f)&=\min_{ g \neq 0} \max_{t\in \CC}\frac{\sum_{(\ell_1,\ell_2)\in E}w_{(\ell_1,\ell_2)}(g(\ell_1)-g(\ell_2))^2}{\sum_{\ell \in V} w_\ell (g(\ell)-t)^2}\geq
	\frac{s^2}{2B^2t^2} \min_{ g \neq 0} \max_{t\in \CC}\frac{\sum_{(\ell_1,\ell_2)\in E}(g(\ell_1)-g(\ell_2))^2}{\sum_{\ell \in V}  (g(\ell)-t)^2} \\
	&=\frac{s^2}{2B^2t^2}\widehat \lambda_G. 
	\end{align*}

\subsection{Proof of Proposition \ref{prop:optimal}}\label{proof:optimal}
Throughout the proof we use the notation defined in Example~\ref{ex:finite}.
	Without loss of  generality, we assume $t=1$ in \eqref{def:Bst}. For $\ell \in [L]$,  let $P_\ell$ denote the projection onto $\H_\ell$.   We choose $f,g \in \B_{s,t}$ by $f={\bf 1}$ and $g(k)=q_{k,d}$, where $q_{k,d} $ denotes the $k$th root of unity of order $d$,  
	$$q_{k,d}=\exp\left({\frac{2\pi i k}{d}}\right). $$
	We claim that for this choice of $f,g$, the inequality 
	\begin{equation}\label{eq:LB}
	\min_{\xi, |\xi|=1}  \|\Phi(f)-\xi \Phi(g)\|_2^2\geq c_{s, t}(\CC,L)^2  \| |\Phi(f)|-|\Phi(g)|\|_2^2  
	\end{equation}
	holds. 
	We first  show that $\|\Phi(f)-\xi \Phi(g)\|_2^2, |\xi|=1$  is proportional to $L$. For  any $\xi$ satisfying $|\xi|=1$,    we have 
	\begin{align} \label{eq:lhs}
	 \|\Phi(f)-\xi \Phi(g)\|_2^2&= \frac{1}{2} \sum_{\ell \in [L]} \|\Phi_\ell(f)-\xi \Phi_\ell(g)\|_2^2 \geq  \frac{A^2}{2}  \sum_{\ell \in [L]}\|P_\ell(f-\xi g)\|_{\H_\ell}^2  \nonumber \\
	&\geq  \frac{A^2}{2}\|f-\xi g\|_{\H}^2= A^2 \sum_{k\in[d]} |1-\xi q_{k,d}|^2=\frac{A^2}{2}\left(2d+ \xi \sum_{k\in[d]}  q_{k,d}+\bar \xi \sum_{k\in[d]}  \bar{q}_{k,d} \right) \nonumber\\
	&=dA^2=aLA^2,  
	\end{align}
	where we use the fact that the sum over all roots of unity of order $d$ is zero. 
	
	Next we show  that $ \| |\Phi(f)|-|\Phi(g)|\|_2^2  $ is proportional to $L^{-1}$.
	For $\ell \in [L]$, let $I_\ell$ be the set of $2a$ consecutive (in the cyclic sense) indices on which the measurements in $\Phi_\ell$ are supported,  and denote $\xi_\ell=q_{k(\ell),d}$ where  
	$k(\ell)$ is some index in $I_\ell$. Then
		\begin{align*}
	 \| \, |\Phi(f)|-|\Phi(g)| \, \|_2^2 &=  \frac{1}{2} \sum_{\ell=0}^{L-1}  \| \,|\Phi_\ell(f)|-|\xi_\ell \Phi_\ell( g)| \,\|_2^2  \leq  \frac{1}{2} \sum_{\ell=0}^{L-1} \|  \Phi_\ell(f)-\xi_\ell \Phi_\ell( g)\|_2^2 \\
	& \leq \frac{B^2}{2}  \sum_{\ell=0}^{L-1}\left( \sum_{j\in I_\ell} |f(j)-\xi_\ell  g(j)|^2  \right) \leq \frac{aB^{2}L}{2} |1-q_{2a,d}|^2\\
	&=a B^{2}L  \left(1-\cos\left(\frac{4\pi a}{d}\right) \right).
	\end{align*}
	Combining this with \eqref{eq:lhs},  we obtain inequality   \eqref{eq:LB} with 
	$$c_{s,t}(\CC,L)=\frac{A}{B}  \left(1-\cos\left(\frac{4\pi a}{d}\right) \right)^{-1/2},$$
	thus concluding the proof of the proposition.

\subsection{Proof of Proposition  \ref{prop:example}}\label{sub:cheeger_infinite}
	 Throughout this proof,  we use the notation defined in Example~\ref{ex:sis}.
	 \begin{enumerate} 
		\item[Part I:] Assume $f(x)=\sum_{k\in \ZZ} c_kB_N(x-k) $ where $|c_k|^p=e^{-\beta|k|} $ for some $\beta>0$. Our goal is to prove that the Cheeger constant of $f$ is positive.  We note that when considering the infimum in the definition of the Cheeger constant in \eqref{eq:cheeger}, it is sufficient to consider  only connected subsets $S$. In the line graph $G$,  connected subsets of $V=\ZZ$ are either finite intervals of the form 
		$$[k,\ell]_\ZZ:=[k,\ell] \cap \ZZ $$
		or one-sided infinite intervals. A simple limiting argument shows that in fact it is sufficient to consider finite intervals only.  Thus  the Cheeger constant in \eqref{eq:cheeger} is reduced to 
		\begin{equation}\label{eq:cheegersis}C_G(f)=\inf_{k<\ell}  \frac{w_{k,k-1}+w_{\ell,\ell+1}}{\min \{\sum_{j \in [k,\ell]_\ZZ} w_j,  \sum_{s \not \in [k,\ell]_\ZZ} w_s\}}.
		\end{equation}
		Note that $w_{k,k-1} $ and $w_{\ell,\ell+1}$ are a summation of $N-1$ entries of the geometric series $\c$, which include the $k$-th and $\ell$-th entries respectively. Thus 
		\begin{equation}\label{eq:edge_bound}
		w_{k,k-1}+w_{\ell,\ell+1} \geq (N-1)e^{-\beta N} (e^{-\beta |k|}+e^{-\beta |\ell|}). 
		\end{equation}
		As for all $j\in \ZZ$, we know that  $\Phi_j$ is a frame on  $\B_j$ with frame constants $(A, B)$ independent of $j$. Explicitly, for all $j\in \ZZ$,
		\begin{equation}\label{sis:frame}
		A^p\sum_{k=j-N+1}^j c_k^p\le \|\Phi_j(f)\|^p\le B^p \sum_{k=j-N+1}^j c_k^p. 
		\end{equation}
		Thus for each $j\in \ZZ$, 
		$$w_j=\|\Phi_j(f)\|^p\le B^p \sum_{k=j-N+1}^j c_k^p= B^p \sum_{k=j-N+1}^j e^{-|k|\beta} \leq NB^p e^{N\beta}e^{-|j|\beta}. $$
		Therefore 
		\begin{equation}\label{eq:vertex_bound}
		\min \{\sum_{j \in [k,\ell]_\ZZ} w_j,  \sum_{s \not \in [k,\ell]_\ZZ} w_s\}\leq NB^pe^{N\beta} 	\min \{\sum_{j \in [k,\ell]_\ZZ} e^{-|j|\beta},  \sum_{s \not \in [k,\ell]_\ZZ} e^{-|s|\beta}\}. 
		\end{equation}
	If $0\leq k<\ell $,  then
		$$\min \{\sum_{j \in [k,\ell]_\ZZ} e^{-|j|\beta},  \sum_{s \not \in [k,\ell]_\ZZ} e^{-|s|\beta}\} \leq \sum_{j=k}^\infty e^{-j\beta}=\frac{e^{-\beta k}}{1-e^{-\beta}} \leq \frac{e^{-\beta|k|}+e^{-\beta|\ell|}}{1-e^{-\beta}}, $$
		and the same inequality can be obtained for $k<\ell \leq 0 $ using the same argument. 
		
			If $k<0<\ell$,  then 
		\begin{align*}\min \{\sum_{j \in [k,\ell]_\ZZ} e^{-|j|\beta},  \sum_{s \not \in [k,\ell]_\ZZ} e^{-|s|\beta}\}& \leq \sum_{s \not \in [k,\ell]_\ZZ} e^{-|s|\beta}=\frac{(e^{-\beta(|k|+1)}+e^{-\beta(\ell+1)})}{1-e^{-\beta}} \\
		&\leq \frac{e^{-\beta|k|}+e^{-\beta|\ell|}}{1-e^{-\beta}}.\end{align*}
		
		Thus returning to \eqref{eq:edge_bound},  \eqref{eq:vertex_bound} and \eqref{eq:cheegersis}, we obtain
		$$C_G(f)\geq \frac{N-1}{NB^p}(1-e^{-\beta})e^{-2N\beta}>0. $$

		\item[Part II:]   We assume that there exists some $\beta>1$ such that 
		$|c_k|^p=(1+|k|)^{-\beta}$. We are going to show the Cheeger constant in \eqref{eq:cheegersis} is zero, i.e.,  $C_G(f)=0$. For $\ell \in \ZZ$,  denote
		$$[\ell,\infty)_\ZZ=[\ell,\infty)\cap \ZZ,  $$
		and choose $k_0>N$ such that $\sum_{k \in [k_0,\infty)_\ZZ} w_k \leq 1/2\sum_{k \in \ZZ} w_k $. Then we have
		\begin{equation}\label{eq:cheeger_zero}
		C_G(f) \leq \inf_{k\geq k_0} \frac{w_{k,k-1}}{\sum_{s \in [k,\infty)_\ZZ} w_s}. 
		\end{equation}
		Set $h(x)=x^{-\beta}$.  For $k \geq 1$, it is easily to verify that 
		$$ h(k)/2 \leq |c_k|^p=(1+|k|)^{-\beta}\leq h(k). $$
		Then we have 
		\begin{equation}\label{eq:poly_edges} w_{k,k-1}=  \sum_{j=k-N+1}^{k-1} |c_j|^p \leq N|c_{k-N+1}|^p\leq Nh(k-N+1)
		\end{equation}
		For all $s\ge k_0 \geq N$, 
		\begin{equation}\label{eq:poly_psi}
	\sum_{j=s-N+1}^{s} |c_j|^p \geq \sum_{j=s-N+1}^s \frac{h(j)}{2} \geq \frac12 \int_{s-N+1}^{s+1}h(x)dx.
		\end{equation}
		By \eqref{sis:frame}, we also  know that 
		$$ w_s=\|\Phi_s(f)\|_p^p\geq A^p\sum_{j=s-N+1}^{s} |c_j|^p , \ \forall  s\in \ZZ. $$
	Using this inequality together with \eqref{eq:poly_psi}, for  $k\ge k_0$,  we obtain
		\begin{align}\label{eq:poly_vert} \sum_{s \in [k,\infty)_\ZZ} w_s & \geq \frac{A^p}{2} \sum_{s \in [k,\infty)_\ZZ} \int_{s-N+1}^{s+1}h(x)dx\nonumber\\
		&  \geq \frac{A^p}{2} \int_{k-N+1}^{\infty} h(x)dx=\frac{A^p}{2(\beta-1)}(k-N+1)^{-(\beta-1)}. 
		\end{align}
		Returning to \eqref{eq:cheeger_zero}, using the bounds from \eqref{eq:poly_edges} and \eqref{eq:poly_vert},  we obtain that 
		$$C_G(f) \leq \frac{2N(\beta-1)}{A^p}\inf_{k \geq k_0} (k-N+1)^{-1}= 0. $$
	 \end{enumerate}

\subsection{Proof of Corollary \ref{cor:sis}} \label{sub:cor}
	Write $f=\sum_{k\in \ZZ^d} c_kB_N(\cdot-k)$ and $g=\sum_{k\in \ZZ^d} d_kB_N(\cdot-k)$. 	By \cite{cheng2019phaseless}, we know any real phase retrieval signal $f\in V(B_N)$ can be determined, up to a sign, from its phaseless samples taken on $\Gamma+\ZZ^d$.  Set $\Gamma'=\{\gamma\subset \Gamma| \ {\rm sign}(f(\gamma)g(\gamma))=1 \ \}$, then we have 
	$$\||f|-|g|\|_{l^2(\Gamma)}^2=\|f-g\|_{l^2(\Gamma')}^2+ \|f+g\|_{l^2(\Gamma\setminus \Gamma')}^2.$$
	By the definition of $\sigma$ in \eqref{def:sigma}, we have 
	\begin{eqnarray}
\sigma &\le& \max\big(\frac{\big(\sum_{\gamma\in \Gamma'}|\sum_{k\in K_0} (c_k-d_k) B_N(\gamma-k)|^2\big)^{1/2}}{(\sum_{k\in K_0}|c_k-d_k|^2)^{1/2}}, \frac{\big(\sum_{\gamma\in \Gamma\setminus\Gamma'}|\sum_{k\in K_0} (c_k+d_k) B_N(\gamma-k)|^2\big)^{1/2}}{(\sum_{k\in K_0}|c_k+d_k|^2)^{1/2}}\big)\nonumber\\
&\le& \frac{\big(\sum_{\gamma\in \Gamma'}|\sum_{k\in K_0} (c_k-d_k) B_N(\gamma-k)|^2\big)^{1/2} + \big(\sum_{\gamma\in \Gamma\setminus\Gamma'}|\sum_{k\in K_0} (c_k+d_k) B_N(\gamma-k)|^2\big)^{1/2}}{\min\big((\sum_{k\in K_0}|c_k-d_k|^2)^{1/2}, (\sum_{k\in K_0}|c_k+d_k|^2)^{1/2}\big)}\nonumber\\
&\le& \frac{\sqrt{2}\big(\sum_{\gamma\in \Gamma}||f(\gamma)|- |g(\gamma)||^2\big)^{1/2}}{\min\big((\sum_{k\in K_0}|c_k-d_k|^2)^{1/2}, (\sum_{k\in K_0}|c_k+d_k|^2)^{1/2}\big)}. \nonumber
\end{eqnarray}
Thus, for $\sigma>0$, we have 
	\begin{equation}\label{cor:eq1}\min_{\xi_0\in \{\pm 1\}}\sum_{k\in K_0}|c_k-\xi_0 d_k|^2 \le 2\sigma^{-2}\sum_{\gamma \in \Gamma}|\ |f(\gamma)|- |g(\gamma)| \ |^2. \end{equation}
Also, we have 
\begin{align*}\min_{\xi_0\in \{\pm 1\}}\sum_{\gamma\in \Gamma}|f(\gamma)-\xi_0 g(\gamma)|^2  &\le\min_{\xi_0\in \{\pm 1\}}\big(\sum_{k\in K_0}|c_k-\xi_0 d_k|^2\big)\sum_{\gamma\in \Gamma}\sum_{k\in K_0}|B_N(\gamma-k)|^2\\
&=\Big(\min_{\xi_0\in \{\pm 1\}}\big(\sum_{k\in K_0}|c_k-\xi_0 d_k|^2\big)\Big)trace({{\bf B}_\Gamma}^T{\bf B}_\Gamma)\\
&\le \Big(\min_{\xi_k\in \{\pm 1\}}\big(\sum_{k\in K_0}|c_k-\xi_k d_k|^2\big)\Big) \sharp(K_0) \lambda^2_{\max} ({\bf B}_\Gamma)\\
&\le 2 \sharp K_0 \lambda^2_{\max} ({\bf B}_\Gamma) \sigma^{-2}\sum_{\gamma \in \Gamma}|\ |f(\gamma)|- |g(\gamma)| \ |^2, \end{align*}
where ${\bf B}_\Gamma=(B_N(\gamma-k))_{\gamma\in \Gamma, k\in K_0}$. Combine with the shift-invariant property of the space $V(B_N)$ implies the $C_0$ in \eqref{const:local} is  $\sqrt{2}(\sharp K_0)^{1/2}\sigma^{-1} \lambda_{\max} ({\bf B}_\Gamma) $ for $p=2$. Obviously, by the construction of $\Psi$, we know $C_1=\sqrt{2}\sigma^{-1}$ in \eqref{eq:edges_dominated} by setting $g=0$ in \eqref{cor:eq1}. 
By Theorem \ref{thm:real}, we have $C_{V(B_N)}=\max\{2(\sharp K_0)^{1/2} \sigma^{-1} \lambda_{\max} ({\bf B}_\Gamma), 4 \sqrt{2}\sigma^{-2}(\sharp K_0)\lambda_{\max} ({\bf B}_\Gamma)\}$.

\bibliographystyle{plain}
\bibliography{bibPR}

\appendix

\section{The Cheeger inequality for infinite graphs}\label{app:Cheeger}
For the completeness of this paper,  we include the proof of  the Cheeger inequality \eqref{ineq:cheeger}, which  holds for the family of infinite graphs we discuss in this paper. Let $G=(V,E, w) $ be a weighted graph with the vertex set $V$ being  countable, the edge set $E \subseteq V \times V$, and the weights $w=(w_v)_{v\in V} \cup (w_{u, v})_{(u, v)\in E}$ associated with  positive numbers $w_{{u,v}}, ({u,v})\in E $ and $w_v, v \in V$. The sequences $(w_{{u,v}})_{({u,v})\in E}$ and $(w_{v})_{v\in V}$ are summable.  We assume that the degree of $G$ is bounded by some $D_N>0 $,  in the sense that for all $v \in V$,
\begin{equation}\label{eq:norm}\sum_{u| ({u,v})\in E } \frac{w_{{u,v}}}{w_v}\leq D_N. 
\end{equation}

Give a graph $G=(V, E,w)$, we define 
$$\vol(S)=\sum_{v \in S} w_v \text{ and } |T|=\sum_{({u,v}) \in T}w_{{u,v}} $$
 for every subset $S \subset V $ and $T \subset E$.
Recall the definition of the Cheeger constant from \eqref{eq:cheeger} as

\begin{equation}\label{eq:cheeger_again}
 C_G=\inf_{ \substack{S \subset V,\\  \vol(S) \leq \vol(V)/2}  } C_S, 
\end{equation} where $C_S= \frac{|\partial S|}{\vol(S)}$. 

The Laplacian of $G$ is a bounded linear operator on the space $\ell^2(V,w) $, whose norm is defined by $\|g\|^2_{\ell^2(V,w)}=\sum_{v\in V} w_v|g_v|^2$. For  $g \in\ell^2(V,w) $ we denote
\begin{equation}\label{def:L}
\L_G(g)=\sum_{({u,v})\in E}(g_u-g_v)^2 w_{{u,v}}, 
\end{equation}
and if additionally $g\neq 0$ we write
\begin{equation}\label{def:R}
R(g)=\frac{L_G(g)}{\|g\|^2_{\ell^2(V,w)}}.
\end{equation} 
Recall the definition of the algebraic connectivity $\lambda_G $ from \eqref{eq:lambda}, we have 
\begin{equation}\label{eq:lambda_again}\lambda_G=\inf_{g \neq 0, g \perp 1} R(g).  \end{equation}

 We then prove the Cheeger inequality by essentially following the proof in \cite{chung2007four}.

\begin{theorem}[Cheeger inequality] Let $G=(V,E, w)$ be a weighted graph. Then the following holds
$$2C_G \geq \lambda_G \geq \frac{C_G^2}{2D_N}.$$
\end{theorem}
\begin{proof}

For every $S\subset V$ satisfying $\vol(S)\leq  \vol(V)/2 $,  we define a function $g_S=(g_v)_{v\in V} \in \ell^2(V, w)$ as 
\begin{equation*}
g_v=\twopartdef{1-\frac{\vol(S)}{\vol(V)}}{v\in S, }{-\frac{\vol(S)}{\vol(V)}}{ v\not\in S.}
\end{equation*}  
This function is orthogonal to $\one$ in $\ell^2(V,w)$, and satisfies
$$\lambda_G \leq \frac{L(g_S)}{\|g_S\|^2_{\ell^2(V,w)}}=\frac{|\partial S |}{\vol(S)(1-\frac{\vol(S)}{\vol(V)})}\leq 2 \frac{|\partial S|}{\vol(S)}=2C_S .$$
Since this inequality holds for all $S\subset V$ satisfying  $\vol(S)\leq \vol(V)/2 $,  we finish the proof of  the first part of the inequality $\lambda_G \leq 2C_G $.

For the other direction, choose any $g=(g_v)_{v\in V}$ orthogonal to $\one$ in $\ell^2(V,w)$. We sort the vertices of $V$ using integer indices $i \in \ZZ$, so that 
$$\ldots \geq g_{v_{i+1}} \geq g_{v_{i}}\geq g_{v_{i-1}} \geq \ldots.   $$
We define the vertex subset $S_i=\{v_j| \ j\leq i\}, i\in \ZZ$ and let $r$ denote the largest integer so that $\vol(S_r) \leq \vol(V)/2 $. Since $g \perp 1 $ we have that
\begin{equation}\label{eq:t}\|g\|_{\ell^2(V,w)}^2=\min_{t \in \RR} \|g-t\one\|_{\ell^2(V,w)}^2 \leq \|g-g_{v_r}\one\|_{\ell^2(V,w)}^2. 
\end{equation}
Write $\tilde{g}=g-g_{v_r}\one$, and define the positive and negative part of $\tilde{g}$ by 
\begin{equation*}
\tilde{g}^+_v=\twopartdefotherwise{g_v-g_{v_r}}{g_v\geq g_{v_r} }{0}{,}
\end{equation*}  
and 
\begin{equation*}
\tilde{g}^-_v=\twopartdefotherwise{|g_v-g_{v_r}|}{g_v\leq g_{v_r} }{0}{.}
\end{equation*}   
Now note that 
\begin{align*}
R(g)=&\frac{\sum_{({u,v}) \in E} (g_u-g_v)^2 w_{{u,v}}}{\|g\|_{\ell^2(V,w)}^2}\stackrel{\eqref{eq:t}}{\geq} \frac{\sum_{({u,v}) \in E} (g_u-g_v)^2 w_{{u,v}}}{\|g-g_{v_r}\one\|_{\ell^2(V,w)}^2}\\
&\geq \frac{\sum_{({u,v}) \in E} \left( [\gg^+_u-\gg^+_v]^2+[\gg^-_u-\gg^-_v]^2 \right)w_{{u,v}} }{\|\gg^+\|_{\ell^2(V,w)}^2+\|\gg^{-}\|_{\ell^2(V,w)}^2 }=\frac{L(\gg^+)+L(\gg^{-})}{\|\gg^+\|_{\ell^2(V,w)}^2+\|\gg^{-}\|_{\ell^2(V,w)}^2}, 
\end{align*}
Without loss of generality, we assume that $R(\gg^+)\leq R(\gg^{-}) $. The last expression can be rearranged so that $R(g)$ is a weighted average of $R(\gg^+) $ and $R(\gg^+) $, and so 
\begin{align*}
R(g)\geq R(\gg^+)&=\frac{\sum_{({u,v})\in E} (\gg^+_u-\gg^+_v)^2w_{{u,v}} }{\|\gg^+\|^2_{\ell^2(V,w)}} 
\frac{\sum_{({u,v})\in E} (\gg^+_u+\gg^+_v)^2w_{{u,v}} }{\sum_{({u,v})\in E} (\gg^+_u+\gg^+_v)^2w_{{u,v}}} \\
&\geq \frac{\left[ \sum_{({u,v})\in E} ((\gg^+_u)^2-(\gg^+_v)^2)w_{{u,v}}
\right]^2 }{2D_N\|\gg^+\|^4_{\ell^2(V,w)}}  \text{ Cauchy-Schwarz inequality in }\ell^2(E,w)\\
&\stackrel{(*)}{=}\frac{\left[\sum_{i < r}((\gg^+_{v_i})^2-(\gg^+_{v_{i+1}})^2)  |\partial S_i|\right]^2}{2D_N\|\gg^+\|^4_{\ell^2(V,w)}}\\
&\geq   \frac{C_G^2}{2D_N}\frac{\left[\sum_{i<r}((\gg^+_{v_i})^2-(\gg^+_{v_{i+1}})^2) \vol (S_i)\right]^2}{\|\gg^+\|^4_{\ell^2(V,w)}}\\
&=\frac{C_G^2}{2D_N}\frac{\left[\sum_{i<r}(\gg^+_{v_i})^2  (\vol (S_i)-\vol(S_{i-1}))\right]^2}{\|\gg^+\|^4_{\ell^2(V,w)}}\\
&=\frac{C_G^2}{2D_N}\frac{\left[\sum_{i<r}(\gg^+_{v_i})^2  w_{v_i}\right]^2}{\|\gg^+\|^4_{\ell^2(V,w)}}=\frac{C_G^2}{2D_N}.
\end{align*}
Here $(*)$ follows from the fact that 
\begin{align*} \sum_{({u,v})\in E} \Big((\gg^+_{u})^2-(\gg^+_{v})^2\Big)w_{{u,v}}&=\sum_{j<k| (v_j,v_k)\in E} ((\gg^+_{v_j})^2-(\gg^+_{v_k})^2)w_{v_j,v_k}\\
&=\sum_{j<k| (v_j,v_k)\in E}\left( \sum_{i=j}^{k-1} (\gg^+_{v_i})^2-(\gg^+_{v_{i+1}})^2   \right) w_{v_j,v_k}\\
&=\sum_i \Big((\gg^+_{v_i})^2-(\gg^+_{v_{i+1}})^2 \Big)  \Big( \sum_{j \leq i<k| (v_j,v_k)\in E} w_{v_j,v_k} \Big)\\
&=\sum_i (( \gg^+_{v_i})^2-(\gg^+_{v_{i+1}})^2) |\partial S_i|. 
\end{align*}
\end{proof}
\end{document}